\documentclass[12pt]{article}

\usepackage{amsfonts,amscd,amssymb}
\usepackage{amsthm,amsmath,natbib}
\usepackage{bbm} 
\usepackage{times,graphicx}
\usepackage{mathrsfs}
\usepackage{enumerate}
\usepackage{subfigure}
\usepackage{appendix}
\usepackage{pdfpages}
\usepackage[margin=1in]{geometry}
\usepackage{setspace}
\usepackage{tikz}
\usepackage{todonotes}

\RequirePackage[colorlinks,citecolor=blue,urlcolor=blue]{hyperref}

\usepackage{natbib}

\newtheorem{thm}{Theorem}
\newtheorem{lemma}{Lemma}

\newtheorem{assumption}{Assumption}

\newcommand{\mb}{\mathbf}
\newcommand{\mc}{\mathcal}

\newcommand{\Ef}{\mathbb{E}}
\newcommand{\E}[1]{\Ef\!\left(#1\right)}
\newcommand{\Pf}{{\mathrm{I}\!\mathrm{P}}}

\newcommand{\tree}{\mathbb{T}}

\newcommand{\var}{\mathrm{Var}}
\DeclareMathOperator*{\argmax}{argmax}

\newcommand{\ol}{\overline}
\newcommand{\ul}{\underline}

\newcommand{\floor}[1]{\lfloor #1 \rfloor}

\title{On the convergence of the maximum likelihood estimator for the transition rate under a $2$-state symmetric model}
\date{}
\author{
Lam Si Tung Ho\thanks{These authors contributed equally to this work.} \\
Department of Mathematics and Statistics \\
Dalhousie University, Halifax, Nova Scotia, Canada
\and
Vu Dinh$^{*}$ \\
Department of Mathematical Sciences\\
University of Delaware
\and
Frederick A.~Matsen IV \\
Program in Computational Biology \\
Fred Hutchinson Cancer Research Center
\and
Marc A.~Suchard \\
Departments of Biomathematics, Biostatistics and Human Genetics \\
University of California, Los Angeles
}

\begin{document}
\maketitle

\clearpage

\begin{abstract}
Maximum likelihood estimators are used extensively to estimate unknown parameters of stochastic trait evolution models on phylogenetic trees.
Although the MLE has been proven to converge to the true value in the independent-sample case, we cannot appeal to this result because trait values of different species are correlated due to shared evolutionary history.
In this paper, we consider a $2$-state symmetric model for a single binary trait and investigate the theoretical properties of the MLE for the transition rate in the large-tree limit.
Here, the large-tree limit is a theoretical scenario where the number of taxa increases to infinity and we can observe the trait values for all species.
Specifically, we prove that the MLE converges to the true value under some regularity conditions.
These conditions ensure that the tree shape is not too irregular, and holds for many practical scenarios such as trees with bounded edges, trees generated from the Yule (pure birth) process, and trees generated from the coalescent point process.
Our result also provides an upper bound for the distance between the MLE and the true value.
\end{abstract}

\section{Introduction}

The \emph{maximum likelihood estimator} (MLE) is frequently used to estimate unknown parameters in trait evolution models on phylogenetic trees.
To safely use this machinery, it is important to know that the MLE is \emph{consistent}: that is, the estimate converges to the true value as we observe trait values for more species.
However, traditional statistical theory only guarantees the consistency of the MLE when observations are independent and identically distributed.
In contrast, trait values of biological species are not independent because species are related to each other according to a phylogenetic tree.
Therefore, traditional consistency results are not directly applicable to trait evolution studies.
For this paper we consider the scenario where only a single trait is observed and the number of species increases to infinity.
This is different from the setting where many traits/characters are observed for the same set of species; this alternate setting typically assumes independence between traits.

Binary traits studied by comparative biologists come in various types, including morphological traits (whether a fruit fly has curly wings), behavioral traits (whether an antelope hides from predators), and geographical traits (whether a species is aquatic).
Here we study the consistency property of the MLE for estimating the transition rate of a binary trait evolved along a phylogenetic tree according to a $2$-state symmetric Markov process.
Under this $2$-state symmetric model, there is an unique rate of switching back and forth between the two states of a phenotype, which is the parameter of interest.
We show that under two mild conditions, the MLE of this transition rate is consistent, that is, the estimate converges to the correct value as we observe the trait values of more species.
These conditions ensure that edge lengths of the tree are not too small and the pairwise distances between leaves are neither too small nor too large.
We verify that these conditions holds for many practical cases including trees with bounded edges, trees generated from the Yule process \citep{yule1925mathematical}, and trees generated from the coalescent point process \citep{lambert2013birth}.

The consistency of the MLE does not always hold under trait evolution models.
Indeed, for estimating the ancestral state of the $2$-state symmetric model, \citet{li2008more} point out that the MLE can be inconsistent.
As a consequence, the estimate of the ancestral state may not be close to the true value no matter how many species have been sampled.
Several efforts have also been made to investigate the consistency of the MLE under evolution models of a continuous trait.
\citet{ane2008analysis} points out that unlike traditional linear regression, the MLE of the coefficients of phylogenetic linear regression under the Brownian motion model can be inconsistent.
Additionally, \citet{sagitov2012interspecies} show that the sample mean is an inconsistent estimator for the ancestral state under this model when the tree is generated from the Yule process.
For the Ornstein-Uhlenbeck model, \citet{ho2013asymptotic} show that if the height of the phylogenetic tree is bounded as more observations are collected, then the MLE of the selective optimum is not consistent.
Moreover, they discovered that in this scenario, no consistent estimator for the selective optimum exists.
Recently, \citet{ane2016phase} provide a necessary and sufficient condition for consistency of the MLE under the Ornstein-Uhlenbeck model.
Although the problem of reconstructing the ancestral state under the $2$-state symmetric model has been studied extensively \citep{tuffley1997links,mooers1999reconstructing,li2008more,mossel2014majority}, it remains unknown whether the transition rate can be estimated consistently.

In this paper, we show that the MLE is a consistent estimator for the transition rate of the $2$-state symmetric model under simple conditions.
We start with introducing the $2$-state symmetric model for binary traits and derive several statistical properties of this model in Section \ref{sec:2state}.
In Section \ref{sec:consistency}, we state two necessary conditions for the consistency of the MLE of the transition rate and provide a detailed proof for this result.
Section \ref{sec:app} verifies these conditions for several practical scenarios and illustrates our result through a simulation.



\section{Properties of the $2$-state symmetric model}
\label{sec:2state}

Let $\mu$ be the transition rate of a $2$-state symmetric Markov process.
Then, the transition probability matrix has an analytical form:
\begin{equation}
\mb{P}_\mu(t) = \begin{pmatrix} \frac{1}{2} + \frac{1}{2}e^{-2 \mu t} & \frac{1}{2} - \frac{1}{2}e^{-2 \mu t} \\ \frac{1}{2} - \frac{1}{2}e^{- 2 \mu t} & \frac{1}{2} + \frac{1}{2}e^{-2 \mu t} \end{pmatrix}.
\label{eqn:trans}
\end{equation}
Hence, the probability that the process switches state after $t$ unit of time is $\frac{1}{2} - \frac{1}{2}e^{- 2 \mu t}$.

In this paper, the term phylogenetic tree (or phylogeny) refers to a bifurcating rooted tree with leaves labeled by a set of taxon (species) names.
We can reroot a phylogenetic tree by moving the root to another location along the tree (see Figure \ref{fig:reroot}).
The evolution of a binary trait along a tree is modeled using the $2$-state symmetric Markov process as follows.
At each node in the phylogeny, the children inherit the trait value of their parent and the trait of each child evolves independently of one another.

\begin{figure}[h]
\centering
\includegraphics[width=0.8\textwidth]{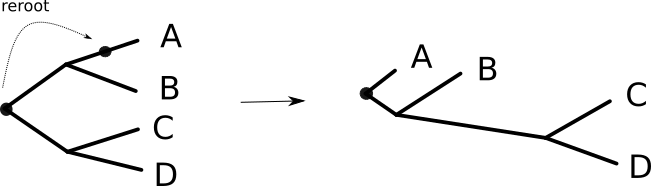}
\caption{Example of rerooting a 4-taxa tree.}
\label{fig:reroot}
\end{figure}

Let $\tree$ be a phylogenetic tree with $n$ leaves, and $\mb{Y}$ be the trait values at the leaves of $\tree$.
We assume that the ancestral state at the root $\rho$ of $\tree$ follows a stationary distribution, which is a Bernoulli distribution with success probability $1/2$.
Let $E$ be the set of all edges of $\tree$.
The joint probability distribution of $\mb{Y}$ is
\[
P_{\tree,\mu}(\mb{Y}) := \Pf(\mb{Y}~|~\tree,\mu) =  \frac{1}{2} \sum_{y}{\left ( \prod_{(u,v)\in E}{[\mb {P}_\mu( d_{uv})}]_{y_u y_v} \right )}
\]
where $y$ ranges over all extensions of $\mb{Y}$ to the internal nodes of the tree, $y_u$ denotes the assigned state of node $u$ by $y$,  $d_{uv}$ is the edge length of $(u,v)$,  and $[\mb{P}_\mu(t)]_{kl}$ is the element at $k$-th row and $l$-th column of matrix $\mb{P}_\mu(t)$.
We define the log-likelihood function as $\ell_{\tree,\mu}(\mb{Y}) = \log P_{\tree,\mu}(\mb{Y})$.
Throughout this paper, we denote the true transition rate with $\mu^*$ and assume that $\mu^* \in [\ul \mu, \ol \mu]$ where $\ul \mu, \ol \mu$ are two known positive numbers.
Define
\[
R_{\tree,\mu}(\mb{Y}) = \ell_{\tree,\mu^*}(\mb{Y}) - \ell_{\tree,\mu}(\mb{Y}).
\]
Then the Kullback-Leibler divergence from $P_{\tree,\mu}$ to $P_{\tree,\mu^*}$ is
\[
\text{KL}(P_{\tree,\mu^*} \| P_{\tree,\mu}) = \Ef [\ell_{\tree,\mu^*}(\mb{Y}) - \ell_{\tree,\mu}(\mb{Y})] = \Ef [R_{\tree,\mu}(\mb{Y})]
\]
where $\Ef$ is the expectation with respect to $P_{\tree,\mu^*}$.

\begin{lemma}
Let  $\tree'$ be the phylogenetic tree obtained by rerooting a tree $\tree$, then
\[
\ell_{\tree',\mu}(\mb{Y}) = \ell_{\tree,\mu}(\mb{Y})
\]
for any trait values $\mb{Y}$ at the leaves.
\label{lem:reroot}
\end{lemma}

This lemma, sometimes called the ``pulley principle,'' is a direct consequence of the fact that the ancestral state follows a stationary distribution \citep{felsenstein1981evolutionary}.
We also have the following Lemmas:

\begin{lemma}
Let $\tree$ be a rooted tree with root $\rho$, and $\mb{Y}$ be the trait values at the leaves of $\tree$ generated under the $2$-state symmetric model.
Let $h$ be a function such that $h(\mb{Y}) = h(\mb{1} - \mb{Y})$, where $\mb{1}$ denotes the all-ones vector.
Then, $h(\mb{Y})$ and the trait value at $\rho$ are independent.
\label{lem:indep}
\end{lemma}

The following two lemmas concern the regularity of the log-likelihood function.
In particular, Lemma \ref{lem:subtrees} shows that the log-likelihood function of a binary tree can be bounded by the sum of the log-likelihood functions of its two subtrees.
\begin{lemma}
Let $\rho_1$ and $\rho_2$ be two direct descendants of $\rho$, and $\tree_1$ and $\tree_2$ the two subtrees descending from them.
Let $\mb{Y}_1$, $\mb{Y}_2$ be the observations at the leaves of $\tree_1$ and $\tree_2$.
We have
\[
|\ell_{\tree,\mu}(\mb{Y}) - \ell_{\tree_1,\mu}(\mb{Y}_1) - \ell_{\tree_2,\mu}(\mb{Y}_2)| \leq \max \left\{\log 2, \log \frac{1}{1 - e^{-2 \mu d}} \right\}.
\]
where $d$ is the tree distance between $\rho_1$ and $\rho_2$.
\label{lem:subtrees}
\end{lemma}

\begin{lemma}[Uniform Lipschitz]
There exists $C>0$ such that
\[
\left|\frac{1}{n}\ell_{\tree, \mu_1}(\mb{Y}) - \frac{1}{n}\ell_{\tree, \mu_2}(\mb{Y}) \right| \le  C |\mu_1 - \mu_2| \qquad \forall \mu_1, \mu_2 \in [\underline{\mu}, \overline{\mu}], \forall \tree
\]
where $n$ denotes the number of leaves of $\tree$.
\label{lem:Lipschitz}
\end{lemma}

The proofs of Lemmas \ref{lem:indep}, \ref{lem:subtrees}, and \ref{lem:Lipschitz} are provided in the Appendix \ref{sec:proof}.
Henceforth we will also use $\rho$ to denote the trait value at the root by an abuse of notation.


\section{Convergence of the MLE of the transition rate}
\label{sec:consistency}

The MLE of $\mu$ is defined as follows:
\[
\hat \mu = \argmax_{\mu \in [\ul \mu, \ol \mu]} \ell_{\tree,\mu}(\mb{Y}).
\]
In this section, we will state our main result regarding the consistency of $\hat \mu$.
We will need to make two assumptions to ensure that the shape of $\tree$ is not too irregular.
Let us define
\[
f(x) = \max \left\{\log 2,\log \frac{1}{1-e^{-2x}}\right\} \text{ and } S_{\tree, \mu}= \sum_{e \in E(\tree)}{f(\mu t_e)^2}.
\]
where $t_e$ denotes the edge length of edge $e$.

\begin{assumption}
$S_{\tree, \mu} = \mc{O}(n^{\gamma}$) for some $1\leq \gamma < 2$.
That is, there exists a universal constant $c < \infty$ such that  $S_{\tree, \mu}< c n^{\gamma} $.
\label{assump:length}
\end{assumption}

\begin{assumption}
There exist $\Omega(n)$ pairs of leaves such that the paths connecting each pair are pairwise disconnected and their lengths are bounded in some fixed range $[\ul d, \ol d]$.
Here, $\Omega(n)$ denotes a quantity that is greater than $c n$ for a positive constant $c$ and all $n$.
\label{assump:contrast}
\end{assumption}

\noindent
Assumption \ref{assump:length} makes sure that the edge lengths of $\tree$ are not too small.
It is worth noticing that $1/(1 - e^{-2x}) \leq 1/x$ when $x$ is small enough.
Therefore, this assumption holds when the smallest edge length is $\Omega(e^{-n (\gamma - 1)/2})$.
On the other hand, Assumption \ref{assump:contrast} guarantees that the pairwise distances between leaves of $\tree$ do not vary too extremely.
In Section \ref{sec:app}, we will verify these assumptions for several common tree models.
Although we employ rerooting in proofs below, Assumption \ref{assump:length} is for the original root of the tree.

\begin{thm}
Under Assumptions \ref{assump:length} and \ref{assump:contrast}, for any $\delta > 0$, $\forall \alpha \in \left ( \max \left \{ \frac{\log(2)}{\log(3/2)}, \gamma \right \}, 2 \right )$, there exists a constant $ C_{\delta,\alpha,\ul d, \ol d, \ul \mu, \ol \mu} > 0$ such that
\[
|\hat \mu - \mu^*| \leq C_{\delta,\alpha,\ul d, \ol d, \ul \mu, \ol \mu} \left( \frac{\sqrt{\log n}}{n^{(2 - \alpha)/6}}\right)
\]
with probability $1-\delta$.
\label{thm:consistency}
\end{thm}

\noindent

From now on, we will use $\ell_{\tree, \mu}$ and $R_{\tree,\mu}$ as short for $\ell_{\tree, \mu}(\mb{Y})$ and $R_{\tree,\mu}(\mb{Y})$.
Recall that $\ell_{\tree, \mu}$ is the log-likelihood function and $R_{\tree,\mu} = \ell_{\tree, \mu^*} - \ell_{\tree, \mu}$ where $\mu^*$ is the true value of the transition rate.
The main ideas of the proof of Theorem \ref{thm:consistency} can be outlined as follows.
For a fixed value of $\mu$, we can view the function $R_{\tree, \mu}$ as the evidence to distinguish between $\mu$ and $\mu^*$.
We will prove that as the number of leaves approaches infinity, the information to distinguish $\mu$ and $\mu^*$, characterized by the KL divergence between $P_{\tree,\mu}$ to $P_{\tree,\mu^*}$ (which is equal to $\Ef(R_{\tree, \mu})$) increases linearly (Lemma \ref{lem:lowerbound}), while the associated uncertainty, characterized by the variance of $R_{T, \mu}$, only increases sub-quadratically (Lemma \ref{lem:boundvar}).
It is worth noting that while the results of Lemma \ref{lem:lowerbound} and Lemma \ref{lem:boundvar} are straightforward for independent and identically distributed data, the analyses for phylogenetic traits are more complicated due to the correlations between trait values at the leaves.
To overcome this issue, we use the independent phylogenetic contrasts, introduced in the next section, to obtain a lower bound on the information.
On the other hand, Lemma \ref{lem:subtrees} shows that the log-likelihood function of a binary tree can be bounded by the sum of the log-likelihood functions of its two subtrees, which allows us to exploit the sparse structure of the tree to derive an upper bound on uncertainty through an induction argument.
Finally, we obtain a uniform bound on the difference between the log-likelihood functions and its expected values (Lemma \ref{lem:concentration}), which enables us to derive an analysis of convergence of the MLE.

\subsection{Lower bound on information}
\label{sec:lower}

\paragraph{Phylogenetic contrasts}
Letting $i$ and $j$ be two different species, we define $\mc{C}_{ij} =Y_i - Y_j$ to be a contrast between the two species.
This is a popular notion introduced by \citet{felsenstein1985phylogenies} for computing the likelihood function under the Brownian motion model.
\citet{ane2016phase} use independent contrasts to construct consistent estimators for the covariance parameters of the Ornstein-Uhlenbeck model.
Here, we will introduce a notion of a squared-contrast $\mc{C}_{ij}^2$, which is simply the square of a contrast $\mc{C}_{ij}$, and show that squared-contrasts have the same independence properties as independent contrasts.
Let $\mc{A}_{ij}$ be the state of the most recent common ancestor of $i$ and $j$, $d_{ij}$ the tree distance from $i$ to $j$, and $p_{ij}$ the path on the tree connecting $i$ and $j$.
By symmetry of the model, we have
\begin{equation}
\Pf(\mc{C}_{ij}^2 = y~|~\mc{A}_{ij} = 0) = \Pf(\mc{C}_{ij}^2 = y~|~\mc{A}_{ij} = 1) =  \begin{cases}
\frac{1}{2}(1 + e^{-2 \mu d_{ij}}), & \mbox{if } y = 0 \\
\frac{1}{2}(1 - e^{-2 \mu d_{ij}}), & \mbox{if } y= 1. \end{cases}
\end{equation}
Therefore, $\mc{C}_{ij}^2$ and $\mc{A}_{ij}$ are independent.

\begin{lemma}[Sequence of contrasts]
Let $\{\mc{C}_{i_k j_k}\}_{k=1}^m$ be a sequence of contrasts such that any two paths in $\{p_{i_k j_k}\}_{k=1}^m$ have no common node.
Then $\{\mc{C}_{i_k j_k}^2\}_{k=1}^m$ are independent.
\label{lem:contrast}
\end{lemma}

\begin{proof}
We have, using $\E{\cdot}$ here to denote expectation over ancestral states,
\begin{align*}
&\Pf(\mc{C}_{i_1 j_1}^2 = y_1, \mc{C}_{i_2 j_2}^2 = y_2, \ldots, \mc{C}_{i_m j_m}^2 = y_m) \\
& = \E{\Pf(\mc{C}_{i_1 j_1}^2 = y_1, \mc{C}_{i_2 j_2}^2 = y_2, \ldots, \mc{C}_{i_m j_m}^2 = y_m | \{\mc{A}_{i_k j_k}\}_{k=1}^m)} \\
& = \E{\prod_{k=1}^m {\Pf\left(\mc{C}_{i_k j_k}^2 = y_k | \{\mc{A}_{i_k j_k}\}_{k=1}^m \right)}} = \E{\prod_{k=1}^m {\Pf\left(\mc{C}_{i_k j_k}^2 = y_k | \mc{A}_{i_k j_k}\right)}} \\
& = \E{\prod_{k=1}^m {\Pf\left(\mc{C}_{i_k j_k}^2 = y_k \right)}} = \prod_{k=1}^m {\Pf\left(\mc{C}_{i_k j_k}^2 = y_k \right)},
\end{align*}
where the third equality comes from the Markov property.
This completes the proof.
\end{proof}

The independent squared-contrasts allow us to derive the following lower bound on the information.

\begin{lemma}
Under Assumption \ref{assump:contrast}, there exists $C_{\ul d, \ol d, \ol \mu}>0$ such that
\[
 \Ef [R_{\tree,\mu}] \ge C_{\ul d, \ol d, \ol \mu} n |\mu - \mu^*|^2
\]
for all $\mu \in [\underline{\mu}, \overline{\mu}]$.
\label{lem:lowerbound}
\end{lemma}

\begin{proof}
Under Assumption \ref{assump:contrast}, there exists a set of $m = \Omega(n)$ pairwise disjoint pairs (that is, the paths connecting each pair are pairwise disconnected) $(Y_{k,1}, Y_{k,2})_{k=1}^m$.
Consider the corresponding set of squared-contrasts $\mc{C}^2_{i_k j_k} = (Y_{k,1} - Y_{k,2})^2$.
By Lemma \ref{lem:contrast}, $(\mc{C}^2_{i_k j_k})_{k=1}^m$ are independent.
Let $Q_{k,\mu}$ be the distribution of $\mc{C}^2_{i_k j_k}$ corresponding to parameter $\mu$.
The total variation from $Q_{k,\mu^*}$ to $Q_{k,\mu}$ is
\begin{align*}
\text{TV}(Q_{k, \mu^*} \| Q_{k, \mu}) &= \frac{1}{2}\sum_{x \in \{ 0,1 \}} |Q_{k,\mu^*}(x) - Q_{k,\mu}(x)| \\
&= | e^{- 2 \mu^* d_{i_k j_k}} - e^{- 2 \mu d_{i_k j_k}} | \geq C_{\ul d, \ol d,\ol \mu} | \mu^* - \mu |.
\end{align*}
By the data processing inequality \citep[Theorem 9 in][]{van2014renyi}, the fact that $(\mc{C}^2_{i_k j_k})_{k=1}^m$ are independent, and Pinsker's inequalities, we have
\begin{align*}
 \Ef [R_{\tree,\mu}] &= \text{KL}(P_{\tree,\mu^*} \| P_{\tree,\mu}) \geq \sum_{k=1}^m{\text{KL}(Q_{k,\mu^*} \| Q_{k,\mu})} \\
 & \geq 2\sum_{k=1}^m{[\text{TV}(Q_{k,\mu^*} \| Q_{k,\mu})]^2} \geq m C_{\ul d, \ol d, \ol \mu} |\mu^* - \mu|^2.
\end{align*}
Note that $m = \Omega(n)$ which completes the proof.
\end{proof}

\subsection{Upper bound on the uncertainty}

The result of Section \ref{sec:lower} indicates that as the number of leaves increases, the information to distinguish $\mu$ and $\mu^*$ also increases linearly.
In order to prove that the MLE can successfully reconstruct the true parameter $\mu$, we need to ensure that such information is not confounded by the uncertainty associated with the log-likelihood function due to randomness in the data.
In other words, we want to guarantee that as $n \to \infty$, the variance of $R_{\tree, \mu}$ only grows sub-quadratically, that is, there exist $\alpha<2$ and $C>0$ such that $\var R_{\tree,\mu} \le C n^\alpha$.
To obtain the bound, we need the following Lemma, the proof of which appears in Appendix~\ref{sec:proof}.

\begin{lemma}
If two functions $u$ and $v$ satisfy $\|u - v\|_{\infty} \le c$, then for all $\omega>1$ and random variables $X$,
\[
\var[u(X)] \le \omega \var[v(X)] + \frac{2\omega}{\omega -1}c^2.
\]
\label{lem:approx}
\end{lemma}
Recalling that,
\[
f(x) = \max \left\{\log 2,\log \frac{1}{1-e^{-2x}}\right\},~~ S_{\tree, \mu}= \sum_{e \in E(\tree)}{f(\mu t_e)^2},
\]
we have the following estimate that does not depend on Assumption \ref{assump:length}.

\begin{lemma}
For all $\alpha \in \left ( \frac{\log(2)}{\log(3/2)}, 2 \right )$ and $\beta \in (0,1)$, there exist $C_{\alpha, \underline{\mu}, \overline{\mu}} , \omega_{\alpha,\beta} > 1$ such that
\[
\var R_{\tree,\mu} \le C_{\alpha, \underline{\mu}, \overline{\mu}} n^\alpha + \frac{8\omega_{\alpha,\beta}}{\omega_{\alpha,\beta} -1} n^{\beta} S_{\tree, \mu}
\]
for all trees $\tree$ with $n \ge 2$ taxa and $\mu \in [\underline{\mu}, \overline{\mu}]$.
\label{lem:boundvar0}
\end{lemma}

\begin{proof}
We will prove the result by induction in the number of taxa $n$.
For $n = 2$, we apply Lemma \ref{lem:Lipschitz} to obtain
\[
R_{\tree,\mu} \leq 2 C | \mu^* - \mu | \leq 2 C | \overline{\mu} - \underline{\mu} |
\]
where $C$ is the constant in Lemma \ref{lem:Lipschitz}.
Thus $ \var R_{\tree,\mu} \leq 4 C^2 | \overline{\mu} - \underline{\mu} |^2$.
Assume that the statement is valid for all $k < n$.
We will prove that it is also valid for $k = n$.
Now, let $\tree$ be a bifurcating tree with $n$ taxa.
\citet{lipton1979separator} show that there exists an edge $e = (I_1, I_2)$ of $\tree$ such that if we reroot $\tree$ at the middle of this edge, the two subtrees $\tree_1$ and $\tree_2$ stemming from $I_1$ and $I_2$ have no more than $2n/3$ leaves.
Let $I$ be the middle point of edge $e$ and $\tree'$ be the tree obtained by rerooting $\tree$ to $I$.
Denote by $n_1$ and $n_2$ the number of leaves of $\tree_1$ and $\tree_2$.
According to Lemma \ref{lem:reroot} and Lemma \ref{lem:indep},
\[
\var R_{\tree,\mu} = \var R_{\tree',\mu} = \var (R_{\tree',\mu}~|~I).
\]
By Lemma \ref{lem:subtrees}, we have $| R_{\tree',\mu} -  R_{\tree_1,\mu} - R_{\tree_2,\mu} | \leq 2f(\mu t_e)$.
We apply Lemma \ref{lem:approx} to obtain
\[
\var R_{\tree,\mu} = \var (R_{\tree',\mu}~|~I)\leq \omega \var (R_{\tree_1,\mu} + R_{\tree_2,\mu} ~|~I) + \frac{8\omega}{\omega -1}f(\mu t_e)^2. \\
\]
Let $\mb{Y}_1$ and $\mb{Y}_2$ be the observations at the leaves of $\tree_1$ and $\tree_2$ respectively.
Since $\mb{Y}_1$ and $\mb{Y}_2$ are independent conditional on $I$, we have $\var (R_{\tree_1,\mu} + R_{\tree_2,\mu} ~|~I) = \var (R_{\tree_1,\mu} ~|~I) + \var (R_{\tree_2,\mu} ~|~I)$.
By Lemma \ref{lem:indep}, we deduce that
\begin{align*}
\var R_{\tree,\mu} & \leq \omega [\var R_{\tree_1,\mu} + \var R_{\tree_2,\mu}] + \frac{8\omega}{\omega -1}f(\mu t_e)^2.
\end{align*}
Recalling that the trees $I_1$ and $I_2$ have no more than $2n/3$ leaves, using the induction hypothesis for $\tree_1$ and $\tree_2$, we have
\begin{align*}
\var R_{\tree,\mu} & \leq \omega C_\alpha [n_1^\alpha + n_2^\alpha] + \frac{8\omega^2}{\omega -1} [n_1^\beta S_{\tree_1} + n_2^\beta S_{\tree_2}]  + \frac{8\omega}{\omega -1} f(\mu t_e)^2 \\
& \leq \omega C_\alpha \frac{2^{\alpha + 1}}{3^\alpha} n^\alpha + \frac{8\omega}{\omega -1} \left [  \omega \frac{2^\beta}{3^\beta} (S_{\tree_1} + S_{\tree_2}) + f(\mu t_e)^2 \right ]n^\beta.
\end{align*}

Let $e_1$ and $e_2$ be the edges stemming from the root of $\tree$.
By definition, we have
\[
S_{\tree} - [S_{\tree_1} + S_{\tree_2} + f(\mu t_e)^2] = f(\mu t_{e_1})^2 + f(\mu t_{e_2})^2 - f[\mu (t_{e_1} + t_{e_2})]^2 \geq 0
\]
where the last inequality comes from the fact that $f$ is a non-increasing function.
Thus, if we choose $\omega$ such that
\[
1<\omega \leq \min \left \{ \left(\frac{3}{2}\right)^{\beta}, \frac{1}{2}\left(\frac{3}{2}\right)^{\alpha} \right \},
\]
then
\[
\var R_{\tree,\mu} \leq C_\alpha n^\alpha + \frac{8\omega}{\omega -1} n^{\beta} S_{\tree},
\]
which completes the proof.
\end{proof}

A combination of Lemma $\ref{lem:boundvar0}$ and Assumption \ref{assump:length} gives rise to the desired bound:

\begin{lemma}[Sub-quadratic upper bound on variance]
Under Assumption \ref{assump:length}, $\forall \alpha \in \left ( \max \left \{ \frac{\log(2)}{\log(3/2)}, \gamma \right \}, 2 \right )$, there exists $C_{\alpha, \gamma} > 0$ such that
\[
\var R_{\tree,\mu} \le C_{\alpha, \gamma} n^\alpha
\]
for all $n \ge 2$ and $\mu \in [\underline{\mu}, \overline{\mu}]$.
\label{lem:boundvar}
\end{lemma}

\begin{proof}
By Lemma \ref{lem:boundvar0} and Assumption \ref{assump:length}, we have
\[
\var R_{\tree,\mu} \le C_\alpha n^\alpha + C_{\alpha, \beta, \gamma} n^{\beta + \gamma},~~\forall \alpha \in \left ( \frac{\log(2)}{\log(3/2)}, 2 \right ), \beta \in (0,1).
\]
Note that $1 \leq \gamma < 2$.
So, $\forall \alpha \in \left ( \max \left \{ \frac{\log(2)}{\log(3/2)}, \gamma \right \}, 2 \right )$, we can choose $\beta = \alpha -\gamma$.
\end{proof}

\subsection{Concentration bound}

We are now ready to prove the concentration inequality for the 2-state symmetric model.

\begin{lemma}[Concentration bound]
Under Assumption \ref{assump:length}, for any $\delta>0$ and $\forall \alpha \in \left ( \max \left \{ \frac{\log(2)}{\log(3/2)}, \gamma \right \}, 2 \right )$, there exists $C_{\delta, \alpha,\ul \mu, \ol \mu}$ such that
\[
\left|\frac{1}{n}R_{\tree,\mu}- \mathbb{E} \left[\frac{1}{n}R_{\tree,\mu} \right]   \right|  \le \frac{C_{\delta,\alpha,\ul \mu, \ol \mu, \gamma}}{n^{(2-\alpha)/3}} \qquad \forall \mu \in [\ul \mu, \ol \mu]
\]
with probability at least $1-\delta$.
\label{lem:concentration}
\end{lemma}

\begin{proof}
Applying Chebyshev's inequality, we obtain
\[
\mathbb{P}\left[ \left|\frac{1}{n}R_{\tree, \mu}- \mathbb{E} \left[\frac{1}{n}R_{\tree, \mu} \right]   \right|  \ge \eta \right] \le \frac{\var(R_{\tree, \mu})}{n^2 \eta^2}
\]
for any $\mu \in [\ul \mu, \ol \mu]$.
On the other hand, by Lemma \ref{lem:Lipschitz}, we have
\[
\left|\frac{1}{n}R_{\tree, \mu_1} - \frac{1}{n}R_{\tree, \mu_2} \right| =  \left|\frac{1}{n}\ell_{\tree, \mu_1} - \frac{1}{n}\ell_{\tree, \mu_2} \right| \le  C_{\ul \mu, \ol \mu} |\mu_1 - \mu_2|, ~~ \forall \mu_1, \mu_2 \in [\ul \mu, \ol \mu].
\]
Therefore, if
\[
\left|\frac{1}{n}R_{\tree, \mu_0}- \mathbb{E} \left[\frac{1}{n}R_{\tree, \mu_0} \right]   \right|  \ge \eta,
\]
then
\[
\left|\frac{1}{n}R_{\tree, \mu}- \mathbb{E} \left[\frac{1}{n}R_{\tree, \mu} \right]   \right|  \ge \frac{\eta}{2},
\]
for all
\[
\mu \in \left [\mu_0 - \frac{\eta}{4 C_{\ul \mu, \ol \mu}}, \mu_0 + \frac{\eta}{4 C_{\ul \mu, \ol \mu}} \right ] \bigcap [\ul \mu, \ol \mu].
\]
Define
\[
\mu_k = \ul \mu + k\frac{\eta}{4 C_{\ul \mu, \ol \mu}},~~ k = 1, 2, \ldots, \left \lfloor  \frac{4 C_{\ul \mu, \ol \mu} (\ol \mu - \ul \mu)}{\eta} \right \rfloor.
\]
We have
\begin{align*}
& \Pf \left[ \exists \mu \in [\ul \mu, \ol \mu]: \left|\frac{1}{n}R_{\tree, \mu}- \mathbb{E} \left[\frac{1}{n}R_{\tree, \mu} \right ]   \right|  \ge \eta \right ] \\
& \qquad \leq \Pf \left [ \bigcup_k \left \{ \left|\frac{1}{n}R_{\tree, \mu_k}- \mathbb{E} \left[\frac{1}{n}R_{\tree, \mu_k} \right ] \right |  \ge \frac{\eta}{2} \right \} \right ]\\
& \qquad \leq \sum_k {\Pf \left [ \left|\frac{1}{n}R_{\tree, \mu_k}- \mathbb{E} \left[\frac{1}{n}R_{\tree, \mu_k} \right ] \right |  \ge \frac{\eta}{2} \right ]} \\
& \qquad \le \sum_k{\frac{\var(R_{\tree, \mu_k})}{n^2 \eta^2}} \\
& \qquad \leq \frac{8 C_{\ul \mu, \ol \mu} (\ol \mu - \ul \mu) C_{\alpha, \gamma} n^{\alpha-2}}{\eta^3}.
\end{align*}
The last inequality comes from Lemma \ref{lem:boundvar}.
We complete the proof by picking
\[
\eta = \left(\frac{8 C_{\ul \mu, \ol \mu} (\ol \mu - \ul \mu) C_{\alpha, \gamma}}{\delta n^{2-\alpha}} \right)^{1/3}.
\]
\end{proof}

\subsection{Proof of Theorem \ref{thm:consistency}}

From Lemma \ref{lem:lowerbound}, we have
\[
C_{\ul d, \ol d, \ol \mu} n |\hat \mu - \mu^*|^2 \leq  \Ef_{\mu^*} [R_{\tree, \hat \mu}] = \Ef_{\mu^*} \left[\ell_{\tree,\mu^*} \right] - \Ef_{\mu^*}  \left[\ell_{\tree,\hat \mu} \right]
\]
where $E_{\mu^*}$ is the expectation with respect to $\bf P_{\tree, \mu^*}$.

Note that
\[
\frac{1}{n} \ell_{\tree, \hat \mu} - \frac{1}{n} \ell_{\tree, \mu^*}   \ge 0.
\]
By Lemma \ref{lem:concentration}, with probability $1 - \delta$, we obtain
\begin{align*}
C_{\ul d, \ol d, \ol \mu} |\hat \mu - \mu^*|^2 & \leq  \left ( \frac{1}{n} \ell_{\tree,\hat \mu} - \Ef_{\mu^*}  \left[\frac{1}{n} \ell_{\tree,\hat \mu} \right] \right) - \left ( \frac{1}{n} \ell_{\tree,\mu^*} - \Ef_{\mu^*}  \left[\frac{1}{n} \ell_{\tree, \mu^*} \right] \right) \\
& \le  \frac{C_{\delta,\alpha,\ul \mu, \ol \mu, \gamma}}{n^{(2-\alpha)/3}},
\end{align*}
which completes the proof.


\section{Applications}
\label{sec:app}

In this section, we discuss applications of Theorem \ref{thm:consistency} for several practical scenarios.
We first consider trees with edges of bounded length.

\begin{thm}[Trees with bounded edges]
If edge lengths of $\tree$ are bounded from below and above, then the MLE is consistent.
\label{thm:boundededges}
\end{thm}

\begin{proof}

Since all edge lengths are bounded from below, Assumption \ref{assump:length} holds with $\gamma = 1$.
Next, we will show that Assumption \ref{assump:contrast} is satisfied.
We select $\floor{n/2}$ pairs of leaves using the following procedure:
\begin{enumerate}
\item Pick a cherry as a pair. \label{step1}
\item Remove the cherry from the tree as well as the edge immediately above it. \label{step2}
\item Repeat step \ref{step1} and \ref{step2} until the tree has $0$ or $1$ leaves.
\end{enumerate}
Denote $\mc{C} = (i_k,j_k)_{k=1}^m$ be the set of pairs returned by the procedure where $m = \floor{n/2}$.
It is obvious that the paths connecting each pair are pairwise disjoint.
Let $\ul e$ and $\ol e$ be the lower and upper bound of edge lengths,
so we have a lower bound for distances between each pair $d_{i_k j_k} \geq 2 \ul e$.
Moreover,
\begin{equation}
\sum_{k=1}^m{d_{i_k j_k}} \leq (2n - 2) \ol e.
\label{eqn:lower}
\end{equation}
We will prove by contradiction that the set $\mc{B} = \{ k: d_{i_k j_k} \leq 8 \ol e \}$ has at least $n/8$ elements.
Assume that $|\mc{B}| < n/8$.
Note that
\[
|\mc{C}| = \floor{n/2} \geq \frac{n-1}{2}.
\]
We deduce that
\[
|\mc{C} \setminus \mc{B}| > \frac{n-1}{2} - \frac{n}{8} = \frac{3n-4}{8}.
\]
Therefore
\[
\sum_{k=1}^m{d_{i_k j_k}} \geq  \sum_{k \in \mc{C} \setminus \mc{B}}{d_{i_k j_k}} > 8 \ol e \frac{3n-4}{8} = (3n-4) \ol e \geq (2n - 2) \ol e, ~~ \forall n \geq 2
\]
which contradicts \eqref{eqn:lower}.
Therefore, $|\mc{B}| \geq n/8$.

\end{proof}

Next, we apply our result to trees generated from a pure-birth process \citep{yule1925mathematical}.
This is a classical tree-generating process which assumes that  lineages give birth independently to one another with the same birth rate.

\begin{thm}[Yule process]
If $\tree$ is generated from a Yule process, then the MLE is consistent.
\label{thm:Yule}
\end{thm}

\begin{proof}

Again, we only need to check Assumptions \ref{assump:length} and \ref{assump:contrast}.
Let $(t_k)_{k=1}^n$ be the amount of time during which $\tree$ has $k$ lineages.
Then, $(t_k)_{k=1}^n$ are independent exponential random variables with rate $k\lambda$ where $\lambda$ is the birth rate of the Yule process.
Let $e_{\min}$ be the minimum edge length of $\tree$.
It is sufficient for us to show that $e_{\min} \geq e^{-n (\gamma - 1)/2}$.
Indeed, we have
\begin{align*}
\Pf \left(e_{\min} \geq e^{-n (\gamma - 1)/2}\right) & = \Pf \left(\min_k {t_k} \geq e^{-n (\gamma - 1)/2}\right) \\
&= \exp \left (- \frac{n(n+1)}{2} \lambda e^{-n (\gamma - 1)/2} \right )
\end{align*}
which converges to $1$ as $n \to \infty$.
Hence, Assumption \ref{assump:length} is satisfied asymptotically with $\gamma > 1$.

Define the age of an internal node to be its distance to the leaves.
By Corollary 4 of \citet{ane2016phase}, there exist $0 < c_1 < c_2 < + \infty$ such that the probability of there being $\Omega(n)$ internal nodes with age between $c_1$ and $c_2$ is $1 - \mc{O}(n^{-1})$.
Let $\mc{I}$ be the set of internal nodes which have age between $c_1$ and $c_2$, we select $\Omega(n)$ pairs of leaves as follows:
\begin{enumerate}
\item Start with an internal node $i \in \mc{I}$ that has the smallest age. Pick two of its descendants such that $i$ is their most recent common ancestor to form a pair. \label{step1b}
\item Remove $i$ and its parent node from $\mc{I}$. \label{step2b}
\item Repeat step \ref{step1b} and \ref{step2b} until $\mc{I}$ is empty.
\end{enumerate}
This procedure is similar to the procedure for selecting contrasts in \citet[][Lemma 3]{ane2016phase}.
Hence, it selects at least $|\mc{I}|/2$ pairs, which is of order $\Omega(n)$.
Since these internal nodes have age between $c_1$ and $c_2$, Assumption \ref{assump:contrast} is satisfied.

\end{proof}

Finally, we consider the coalescent point process (CPP), which generates a random ultrametric tree with a given height $T$ \citep[see][for more details]{lambert2013birth}.
Conditioning on the number of species $n$,  the internal node ages $(t_k)_{k=1}^{n-1}$ of the tree generated from the CPP are independent and identically distributed according to a probability distribution in $[0,T]$.
The probability density function $\phi$ of this common distribution is called the coalescent density.
We say that a CPP is \emph{regular} if its common distribution is not a point mass at $0$ and the coalescent density is bounded from above.

\begin{thm}[Coalescent point process]
If $\tree$ is generated from a regular coalescent point process, then the MLE is consistent.
\label{thm:CPP}
\end{thm}

\begin{proof}

Denote the minimum edge length of $\tree$ by $e_{\min}$.
We have
\[
e_{\min} \geq \min \left \{ \min_{1 \leq i < j \leq n-1}{|t_i - t_j|}, \min_{1 \leq i \leq n-1}{t_i}, \min_{1 \leq i \leq n-1}{(T - t_i)} \right \} .
\]
Therefore,
\begin{multline*}
\Pf \left (e_{\min} \leq \frac{1}{n^3} \right ) \leq \Pf \left (\min_{1 \leq i < j \leq n-1}{|t_i - t_j|} \leq \frac{1}{n^3} \right ) \\
 + \Pf \left ( \min_{1 \leq i \leq n-1}{t_i} \leq \frac{1}{n^3} \right ) + \Pf \left ( \min_{1 \leq i \leq n-1}(T - t_i) \leq \frac{1}{n^3} \right ).
\end{multline*}
By the results of Section 4 in \citet{jammalamadaka1986limit} (for which the details will be provided in the Appendix), we have
\begin{equation}
\Pf \left (n^2 \min_{1 \leq i < j \leq n-1}{|t_i - t_j|} \leq \epsilon \right ) \to 1 - \exp \left (-c \epsilon \int_0^T{\phi^2} \right )
\label{eqn:jamma}
\end{equation}
for any $\epsilon > 0$. We deduce that
\begin{align*}
\lim_{n \to \infty} \Pf \left (n^2 \min_{1 \leq i < j \leq n-1}{|t_i - t_j|} \leq \frac{1}{n^3} \right ) & \leq \lim_{n \to \infty} \Pf \left (n^2 \min_{1 \leq i < j \leq n-1}{|t_i - t_j|} \leq \epsilon \right )\\
& = 1 - \exp \left (-c \epsilon \int_0^T{\phi^2} \right ).
\end{align*}
Since $\epsilon$ can be arbitrary small, we conclude
\[
\lim_{n \to \infty} \Pf \left (\min_{1 \leq i < j \leq n-1}{|t_i - t_j|} \leq \frac{1}{n^3} \right ) = 0.
\]
On the other hand
\[
\Pf \left ( \min_{1 \leq i \leq n-1}{t_i} \leq \frac{1}{n^3} \right ) = 1 - \Pf \left ( \min_{1 \leq i \leq n-1}{t_i} \geq \frac{1}{n^3} \right ) = 1 - \left ( \int_{1/n^3}^T{\phi} \right )^n.
\]
Note that $\phi$ is bounded from above by $M$ in $[0,T]$, thus
\[
\left ( \int_{1/n^3}^T{\phi} \right )^n = \left ( 1 - \int_{0}^{1/n^3}{\phi} \right )^n \geq \left (1 - \frac{M}{n^3} \right )^n \to 1.
\]
where $M$ is the upper bound of $\phi$.
Hence,
\[
\lim_{n \to \infty} \Pf \left ( \min_{1 \leq i \leq n-1}{t_i} \leq \frac{1}{n^3} \right ) = 0.
\]
Similarly,
\[
\lim_{n \to \infty} \Pf \left ( \min_{1 \leq i \leq n-1}(T - t_i) \leq t \right ) = 0.
\]
Therefore,
\[
\Pf \left (e_{\min} \geq \frac{1}{n^3} \right ) \to 1.
\]
Hence, Assumption \ref{assump:length} is satisfied.

Since the common distribution is not a point mass at $0$, there exist a constant $c > 0$ such that $\int_c^T{\phi} > 0$.
Let $\mc{I}$ be the set of internal nodes of $\tree$ which have age between $c$ and $T$.
Then, Assumption \ref{assump:contrast} holds if we can prove that $|\mc{I}|$ is of order $\Omega(n)$ because we can use the same argument as in the proof of Theorem \ref{thm:Yule}.
By strong law of large numbers, we have
\[
\frac{|\mc{I}|}{n-1} \to \int_c^T{\phi} > 0,
\]
which completes the proof.

\end{proof}


\section{Practical implication}

The consistency property proved in Theorem \ref{thm:consistency} suggests that adding more taxa helps to significantly improve the accuracy of the MLE of the transition rate.
It is worth noticing that this property does not hold in many scenarios \citep[see][]{li2008more, ane2008analysis, ho2013asymptotic, ho2014intrinsic}.
To illustrate our theoretical results, we perform the following simulation using the \texttt{R} package \texttt{geiger} \citep{harmon2007geiger, pennell2014geiger}.
We simulate $100$ trees (the number of taxa varies from $50$ to $2000$) according to the Yule process with birth rate $\lambda = 1$.
For each tree, we simulate $100$ traits under the $2$-state symmetric model with the transition rate $\mu = 0.5$.
For each of these 100 traits, the MLE for the transition rate is computed separately. The result is summarized in Figure \ref{fig:MLE}.
We can see that the MLEs concentrate more and more around the true transition rate as the number of species increases.

\begin{figure}[h]
\centering
\includegraphics[scale = 0.5]{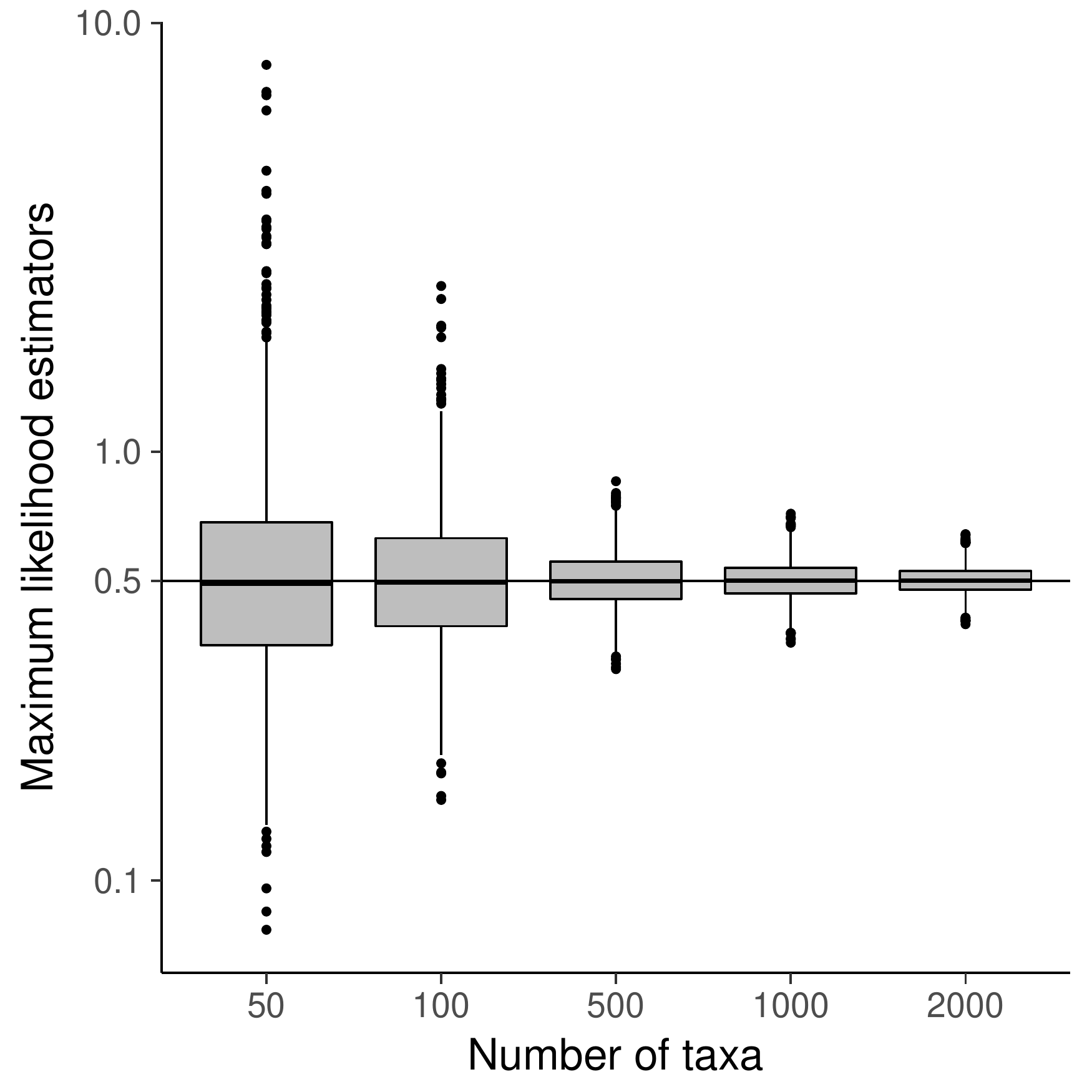}
\caption{Box plots of the MLEs for the transition rate of the $2$-state symmetric model on trees generated from the Yule process with birth rate $\lambda = 1$. The true value of the transition rate is $0.5$.}
\label{fig:MLE}
\end{figure}


\section{Discussion}

In this paper, we investigate the convergence of the MLE for the transition rate of a $2$-state symmetric model under regularity conditions on tree shape.
These conditions ensure that edge lengths of the tree are not too small and the pairwise distances between leaves are not too extreme.
For example, these conditions are satisfied when the edge lengths of the tree are bounded from above and below.
We have also verified that trees generated from pure-birth process and coalescent point process satisfy these conditions.
On the other hand, whether these sufficient conditions are also necessary conditions remains open.

Our results suggest that adding data in terms of the number of tips tends to improve the accuracy of the MLE of the transition rate in general.
To demonstrate this result, we use simulations to confirm that the MLE of the transition rate is consistent when the tree is generated from the Yule process.
It is worth noticing that for evolutionary data, the MLE for other quantities of interest in evolutionary studies can be inconsistent \citep{li2008more, ane2008analysis, ho2013asymptotic, ho2014intrinsic}.
In these situations, adding more data is a waste of resources because it does not significantly improve the precision of the MLE.
Therefore, it is important to study the consistency of the MLE in the context of trait evolution models.

Beside its practical biological application, the paper also investigates many theoretical properties of tree-generating processes.
Specifically, in Theorem $\ref{thm:Yule}$ and Theorem $\ref{thm:CPP}$, we consider the problem of bounding the minimum edge length $e_{min}$ of a tree generated by Yule and coalescent point processes.
The theorems show that generically, $e_{min}$ is bounded from below by a function of order $n^{-2}$ for the Yule process and $n^{-3}$ for the coalescent point process.
Since the minimum edge length plays an important role in many phylogenetic problems, these results may be of independent interest.

Finally, we remark that if we have multiple independent traits which evolve according to the same 2-state Markov process, the distance from the MLE to the true value of the transition rate will decrease linearly with respect to the number of traits.
This result comes from the fact that traditional properties of MLE can be applied because these traits are independent.
Therefore, incorporating additional traits to the analysis is another way to improve the precision of the MLE.


%


\appendix

\section{Proofs}
\label{sec:proof}

\subsection{Proof of lemma \ref{lem:indep}}

Note that by symmetry, we have $\Pf(\mb{Y} = \mb{y}~|~\rho = 0) = \Pf(1 - \mb{Y} =  \mb{y}~|~ \rho = 1)$.
We deduce that
\begin{align*}
\Pf(h(\mb{Y}) = x~|~\rho = 0) &= \Pf(\mb{Y} \in h^{-1}(x)~|~\rho = 0) \\
&= \Pf(\mb{1} - \mb{Y} \in h^{-1}(x)~|~\rho = 1) \\
&= \Pf(h(\mb{1} - \mb{Y}) = x~|~\rho = 1) \\
&= \Pf(h(\mb{Y}) = x~|~\rho = 1)
\end{align*}
which completes the proof.


\subsection{Proof of lemma \ref{lem:subtrees}}

Denote $P^{(u)}_v = \Pf(\mb{Y_u} ~|~ \tree_u, \mu, \rho_u = v)$ for $u \in \{ 0, 1\}$, $v \in \{ 0,1 \}$.
We have
\[
P_{\tree_1,\mu}(\mb{Y_1}) P_{\tree_2,\mu}(\mb{Y_2}) = \frac{1}{4} \sum_{u,v \in \{ 0, 1\}}{P^{(1)}_u P^{(2)}_v}.
\]
Moreover
\[
P_{\tree,\mu}(\mb{Y}) = \frac{1 + e^{-2 \mu d}}{4} \sum_{u \in \{ 0, 1 \}}{P^{(1)}_u P^{(2)}_u} + \frac{1 - e^{-2 \mu d}}{4} \sum_{u \in \{ 0, 1 \}}{P^{(1)}_u P^{(2)}_{1-u}}.
\]
Therefore
\[
\frac{1}{1 - e^{-2 \mu d}} P_{\tree_1,\mu}(\mb{Y_1}) P_{\tree_2,\mu}(\mb{Y_2}) \leq P_{\tree,\mu}(\mb{Y}) \leq 2 P_{\tree_1,\mu}(\mb{Y_1}) P_{\tree_2,\mu}(\mb{Y_2}).
\]


\subsection{Proof of lemma \ref{lem:Lipschitz}}

Without loss of generality, we assume that $\mu_1 < \mu_2$.
By the mean value theorem, there exists $\tilde \mu_{uv} \in (\mu_1, \mu_2)$ for any $u, v \in \{ 0, 1\}$ such that
\[
\left| \log [\mb{P}_{\mu_1}(t)]_{uv} - \log [\mb{P}_{\mu_2}(t)]_{uv} \right| = \frac{ t e^{- 2 \tilde \mu_{uv} t}}{[\mb{P}_{\tilde \mu_{uv}}(t)]_{uv}} |\mu_1 - \mu_2| \leq \frac{ t e^{- 2 \tilde \mu_{uv} t}}{1 - e^{- 2 \tilde \mu_{uv} t}} |\mu_1 - \mu_2|.
\]
We observe that there exists a $C_{\underline{\mu}, \overline{\mu}}>0$ such that
\[
 \sup_{t \ge 0; \tilde \mu_{uv} \in (\underline{\mu}, \overline{\mu})} {\frac{ t e^{- 2 \tilde \mu_{uv} t}}{1 - e^{- 2 \tilde \mu_{uv} t}}} \le C_{\underline{\mu}, \overline{\mu}}.
\]
Therefore,
\[
| \log [\mb{P}_{\mu_1}(t)]_{uv} - \log [\mb{P}_{\mu_2}(t)]_{uv} | \leq C_{\underline{\mu}, \overline{\mu}} |\mu_1 - \mu_2|.
\]
This implies that
\begin{equation}
[\mb{P}_{\mu_1}(t)]_{uv} \leq e^{C_{\underline{\mu}, \overline{\mu}} |\mu_1 - \mu_2|} [\mb{P}_{\mu_2}(t)]_{uv}.
\label{eqn:Pbound}
\end{equation}
Note that
\[
P_{\tree,\mu}(\mb{Y}) =  \frac{1}{2} \sum_{y}{\left ( \prod_{(u,v)\in E}{[\mb {P}_\mu( d_{uv})}]_{y_u y_v} \right )}.
\]
By applying \eqref{eqn:Pbound} for all $2n-3$ edges on the tree, we deduce that
\[
P_{\tree,\mu_1}(\mb{Y}) \leq e^{(2n-3) C_{\underline{\mu}, \overline{\mu}} |\mu_1 - \mu_2| } P_{\tree,\mu_2}(\mb{Y}) .
\]
Hence,
\[
|\ell_{\tree,\mu_1}(\mb{Y}) - \ell_{\tree,\mu_2}(\mb{Y})| \leq (2n-3) C_{\underline{\mu}, \overline{\mu}} |\mu_1 - \mu_2|,
\]
which validates the lemma.


\subsection{Proof of lemma \ref{lem:approx}}

For all $x, y$, we have $|u(x) - u(y)| \le |v(x)-v(y)| + 2c$.
Let $Y$ be an independent and identically distributed copy of $X$, we have
\begin{align*}
2\var[u(X)] &=  \Ef_{X}[u(X)^2] +  \Ef_{Y}[u(Y)^2]  - 2  \Ef_{X}[u(X)]  \Ef_{Y}[u(Y)]\\
&= \Ef_{X, Y}\left(u(X)^2 + u(Y)^2 - 2 u(X) u(Y)\right)\\
&= \Ef_{X, Y}\left([u(X)-u(Y)]^2\right) \\
&\le \Ef_{X, Y}\left( [|v(X)-v(Y)| + 2c]^2 \right).
\end{align*}
Note that for all $z, c \in \mathbb{R}$ and $\omega>1$,
\[
(z + 2c)^2 \le \omega z^2 + \frac{4\omega}{\omega -1} c^2.
\]
Therefore,
\begin{align*}
2\var[u(X)]  & \le \omega \Ef_{X, Y}\left([v(X)-v(Y)]^2 \right) +\frac{4\omega}{\omega -1}c^2\\
& = 2 \omega \var[v(X)] +\frac{4\omega}{\omega -1}c^2.
\end{align*}


\subsection{Proof of Equation \eqref{eqn:jamma}}

In order to establish Equation \eqref{eqn:jamma}, we use the following Lemma.

\begin{lemma}[Remark 3.4 in \citet{jammalamadaka1986limit}]
Let $X_1, X_2, \ldots, X_n$ be an i.i.d. sequence of random variables and $f_n(x, y)$ be an indicator function on $\mathbb{R}^2$ such that
\[
n^3 E[f_n(X_1, X_2) f_n(X_1, X_3)] \to 0 ~~~ \text{and} ~~~\frac{1}{2}n^2E[f_n(X_1, X_2)] \to \lambda
\]
for some constant $\lambda>0$.
Define $U_n =\sum_{1 \le i< j \le n}{f_n(X_i, X_j)}$.

Then $U_n \to_d \text{Poisson}(\lambda)$.
\end{lemma}

We apply this Lemma with $f_n(x, y) = I\{|x-y| < r_n\}$ where $r_n = \epsilon/n^2$ for the sequence $t_1, t_2, \ldots, t_n$ of the coalescent point process.
Note that by Equation (4.3) in \citet{jammalamadaka1986limit},
\[
\frac{1}{2}n^2E[f_n(t_1, t_2)] \to c \epsilon \int_{0}^T{\phi(x)^2 dx}
\]
for some constant $c>0$.
On the other hand, we have
\begin{align*}
E[f_n(t_1, t_2) f_n(t_1, t_3)] &= E[(E[f_n(t_1, t_2) \mid t_1] )^2]\\
&= \int_{0}^T{\left(\int_{t - r_n}^{t+r_n}{\phi(\tau)d\tau}\right)^2\phi(t) dt}\le \frac{4 \| \phi \|_\infty^2 \epsilon^2}{n^4}.
\end{align*}
Therefore, $n^3 E[f_n(t_1, t_2) f_n(t_1, t_3)] \to 0$.
Hence,
\[
\Pf \left (n^2 \min_{1 \leq i < j \leq n-1}{|t_i - t_j|} \leq \epsilon \right ) = P(U_n = 0) \to 1 - \exp \left (-c \epsilon \int_0^T{\phi^2} \right ).
\]


\bibliographystyle{chicago}
\bibliography{ms}

\end{document}